\documentclass[a4paper]{amsart} 
\usepackage{amssymb,amsthm} 
\usepackage[colorlinks,linkcolor=red,anchorcolor=blue,citecolor=green]{hyperref}
\usepackage{tikz}
\usetikzlibrary {arrows.meta,automata,positioning}

\usepackage{enumitem}
\setlist[enumerate]{label*=(\arabic*)}

\usepackage[capitalise]{cleveref}
\Crefname{enumi}{}{}
\crefname{enumi}{}{}

\usepackage[all]{xy}

\newtheorem{theorem}{Theorem}[section]
\newtheorem{conjecture}[theorem]{Conjecture}

\newtheorem{lemma}[theorem]{Lemma}

\newtheorem{remark}[theorem]{Remark}

\newtheorem{proposition}[theorem]{Proposition}

\DeclareMathOperator{\MCA}{\mathcal{A}}                               
\DeclareMathOperator{\MBQ}{\mathbb{Q}}                               
\DeclareMathOperator{\MCK}{\mathcal{K}}                               
\DeclareMathOperator{\dupl}{\mathsf{dupl}}                               
\DeclareMathOperator{\excl}{\mathsf{excl}}                              
\DeclareMathOperator{\rt}{\mathsf{rt}}                               
\DeclareMathOperator{\Span}{\mathsf{span}}                               
\DeclareMathOperator{\kk}{\mathbf{k}}                               
\DeclareMathOperator{\cone}{\mathsf{cone}}                               
\DeclareMathOperator{\chr}{\mathsf{1}}                               
\DeclareMathOperator{\weakc}{\mathsf{wc}}                               
\DeclareMathOperator{\strc}{\mathsf{sc}}                               
\DeclareMathOperator{\id}{\mathsf{id}}                               
\DeclareMathOperator{\Sink}{\mathsf{SinkC}}                               
\DeclareMathOperator{\Strong}{\mathsf{SCC}}                               
\DeclareMathOperator{\Weak}{\mathsf{WCC}}                               
\DeclareMathOperator{\TT}{\mathsf{T}}                               

\DeclareMathOperator{\transLen}{\mathsf{len}}                               
\DeclareMathOperator{\limit}{\mathsf{lim}}                               
\DeclareMathOperator{\Sym}{\mathsf{Sym}}                               
\DeclareMathOperator{\ST}{\mathbb{ST}}                               
\DeclareMathOperator{\dd}{\mathsf{d}}                               
\DeclareMathOperator{\Dim}{\mathsf{dim}}                               

\date{} 

\makeatletter %

\def\theenumi{\@arabic\c@enumi}
\makeatother

\begin{document}

\title{A quadratic upper bound on the reset thresholds of synchronizing automata containing a transitive permutation group}
\author{Yinfeng Zhu}
\begin{abstract}
For any synchronizing $n$-state deterministic automaton, \v{C}ern\'{y} conjectures the existence of a synchronizing word of length at most $(n-1)^2$.
We prove that there exists a synchronizing word of length at most $2n^2 - 7n + 7$ for every synchronizing $n$-state deterministic automaton that satisfies the following two properties: 1. The image of the action of each letter contains at least $n-1$ states; 2. The actions of bijective letters generate a transitive permutation group on the state set.
\end{abstract}


\maketitle

\section{Introduction}

\subsection{Synchronizing automata and \v{C}ern\'{y} Conjecture}

Let $Q$ be a set.
Denote the set of all mappings from $Q$ to itself by $\TT(Q)$.
For the purposes of this article, an \emph{automaton} $\MCA$ is a triple $(Q, \Sigma, \delta)$ where $Q$ and $\Sigma$ are two finite sets, $\delta$ is a mapping from $\Sigma$ to $\TT(Q)$.
The elements of $Q$ are called \emph{states} of $\MCA$; the elements of $\Sigma$ are called letters of $\MCA$; and $\delta$ is called the \emph{transition function} of $\MCA$.
For a mapping $f : X \to Y$ and $x \in X$, we denote the value of $f$ at $x$ by $x.f$ or $f(x)$.
When the transition function $\delta$ is clear from the context, to simplify notations, $q.(\delta(a))$ will be shortened to $q.a$ where $q \in Q$ and $a \in \Sigma$. For subsets $P \subseteq Q$ and $A \subseteq \Sigma$, write $P.A$, or $P.a$ if $A = \{a\}$, for the set $\{p.a : p \in P, a \in A\}$.

Let $X$ be a set. Finite sequences over $X$ (including the empty sequence denoted by $\epsilon$) are called \emph{words}. For each nonnegative integer $i$, write  $X^{i}$ ($X^{\le i}$, respectively) for the set of words of length $i$ (at most $i$, respectively). Denote the set of all words over $X$ by $X^*$.

The transition function $\delta$ extends to the mapping of the set of finte words $\Sigma^*$ on $\TT(Q)$ (still denoted by $\delta$) via the recursion: $q.\epsilon = q$ and $q.(wa) = (q.w).a$ for every $w \in \Sigma^*$, $a \in \Sigma$ and $q \in Q$. 

Let $\MCA = (Q, \Sigma, \delta)$ be an automaton.
A word $w \in \Sigma^*$ is a \emph{reset word} if $|Q.w| = 1$. An automaton that admits a reset word is called a \emph{synchronizing} automaton. The minimum length of reset words for $\MCA$ is called the \emph{reset threshold} of $\MCA$, denoted $\rt(\MCA)$.
The following conjecture is the most famous conjecture of synchronizing automata.
\begin{conjecture}[\v{C}ern\'{y}-Starke] \label{conj:cerny}
	Let $\MCA$ be an $n$-state synchronizing automaton. Then $\rt(\MCA) \le (n-1)^2$.
\end{conjecture}
This conjecture is usually called \v{C}ern\'{y} Conjecture \cite{Cer71}, although it was first published in 1966 by Starke \cite{Sta66}. Regarding the history of \cref{conj:cerny}, we recommend \cite[Section 3.1]{Vol22}.

\v{C}ern\'{y} \cite{Cer64} showed that there exists an $n$-state automaton with the reset threshold equal to $(n-1)^2$ for every $n$. 
That means the upper bound in \cref{conj:cerny} is optimal.

For a long time, the best upper bound of reset thresholds was $\frac{n^3-n}{6}$, obtained by Pin and Frankl \cite{Fra82,Pin83}. In 2018, Szyku{\l}a \cite{Szy18} improved the Pin-Frankl bound. Based on Szyku{\l}a's method, Shitov \cite{Shi19} made a further improvement and obtained the new upper bound $cn^3 + o(n^3)$, where the coefficient $c$ is close to $0.1654$.

Although \v{C}ern\'{y} Conjecture is widely open in general, it has been shown to be true in many special classes e.g. \cite{Dubuc98,Kari01,Tra07,Vol09}. 
For a summary of the state-of-the art around the \v{C}ern\'{y} Conjecture, we recommend the two surveys \cite{KV21,Vol22}.

\subsection{Automata containing transitive groups and our contribution}

Let $\MCA = (Q, \Sigma, \delta)$ be an automaton. 
The \emph{defect} of a word $w \in \Sigma^*$ is the integer $|Q| - |Q.w|$.
For a non-negative integer $i$, write $\Sigma_i$  for the set of letters of defect $i$.

Let $A \subseteq \Sigma_0$. 
Observe that $\delta$ induces a homomorphism from the free monoid $A^*$ to the \emph{symmetric group} $\Sym(Q)$.
We will say that $\MCA$ \emph{contains} the permutation group $\delta(A^*)$ with the generating set $\delta(A)$.
A subgroup $G$ of $\Sym(Q)$ is \emph{transitive} if $q.G = Q$ for each $q \in Q$.
We use $\ST$ to denote the family of synchronizing automata that contain a transitive permutation group on its state set.

In this article, we focus on automata in $\ST$. 
Many subfamilies of $\ST$ have been studied in detail \cite{ACS17,Rys95,Rys00,Rys22,RS24,Rus23,Ste10}. 
We introduce two important results that are strongly relevant to this article.

For any $A \subseteq \Sigma_0$, the minimum integer $d$ such that $\delta(A^{d}) = \delta(A^*)$ is denoted by $\dd_{\MCA}(A)$. Observe that $\dd_{\MCA}(A)$ is the diameter of the Cayley digraph of the group $\delta(A^*)$ with the generating set $\delta(A)$. The following theorem is essentially contained in the results of Rystsov \cite{Rys95}.
\begin{theorem}[Rystsov] \label{thm:Rystsov}
	Let $\MCA = (Q, \Sigma, \delta) \in \ST$ be an $n$-state automaton. 
	Then $\rt(\MCA) \le 1+ (n-2)(n - 1 + \dd_{\MCA}(A))$, where $A \subseteq \Sigma_0$ such that $\delta(A^*)$ is transitive.
\end{theorem}
Ara\'{u}jo, Cameron and Steinberg \cite[Theorem 9.2]{ACS17}, using representation theory over the field of rationals $\mathbb{Q}$, have improved  Rystsov’s bound as displayed in \cref{thm:Steinberg}. 
It is worth mentioning that a similar result can be also found in \cite[Theorem 3.4]{Ste10}.
\begin{theorem}[Ara\'{u}jo-Cameron-Steinberg] \label{thm:Steinberg}
	Let $\MCA = (Q, \Sigma, \delta) \in \ST$ be an $n$-state automaton. 
	Then $\rt(\MCA) \le 1+ (n-2)(n - m + \dd_{\MCA}(A))$, where $A \subseteq \Sigma_0$ such that $\delta(A^*)$ is transitive and $m$ is the maximum dimension of an irreducible $\mathbb{Q}(\delta(A^*))$-module of $\mathbb{Q}^{Q}$
\end{theorem}
In the case that $\dd_{\MCA}(A)$ is small (a linear function of $n$), \cref{thm:Rystsov,thm:Steinberg} bound $\rt(\MCA)$ from above by a quadratic function of $n$, or even verify \v{C}ern\'{y} Conjecture \cite{Rys95,ACS17,Ste10}.

In this article, we obtain \cref{thm:main} which improves \cref{thm:Rystsov} in a different way.
As an application of \cref{thm:main}, we obtain the following result.
\begin{theorem} \label{thm:transitive}
	Let $\MCA \in \ST$ be an $n$-state automaton. If $\Sigma = \Sigma_0 \cup \Sigma_1$, then $\rt(\MCA) \le 2n^2 - 7n + 7$.
\end{theorem}

\subsection{Approach and layout}

To prove \cref{thm:transitive,thm:main}, we use the so-called \emph{extension method} which is based on \cref{prop:ExtensionMethod}. The proof of \cref{prop:ExtensionMethod} can be found in many papers (e.g. \cite[Section 3.4]{Vol22}). 

Let $\MCA = (Q, \Sigma, \delta)$ is an automaton. 
For a subset $S \subseteq Q$ and a word $w \in \Sigma^*$, write $S.w^{-1}$ for the set $\{q \in Q: q.w \in S\}$.
A subset $S \subseteq Q$ is \emph{extended} by a word $w \in \Sigma^*$ if $|S.w^{-1}| > |S|$. A subset $S \subseteq Q$ is called \emph{$m$-extensible}, or \emph{extensible}, if $S$ is extended by a word of length at most $m$.

\begin{proposition} \label{prop:ExtensionMethod}
	Let $\MCA = (Q, \Sigma, \delta)$ be a synchronizing automaton.
	If every nonempty proper subset $S$ of $Q$ is $m$-extensible, then $\rt(\MCA) \le 1 + (n-2)m$. 
\end{proposition}

The remaining of this article will proceed as follows.
In \cref{sec: Linear}, combining the extension method and a dimensional argument for a linear structure, we establish a upper bound for the reset thresholds of automata in $\ST$, see \cref{thm:main}. In \cref{sec: RystsovDigraphs}, using some graph theoretical techniques, we present a proof of \cref{thm:transitive}. Using these graph theoretical techniques,  we can slightly improve some known results about reset thresholds.
At the end, we summarize our results in \cref{sec:conclusion}. 

\section{Linear structure} \label{sec: Linear}

In this section, we will encode some information of an $n$-state synchronizing automaton into some objects in $\MBQ^n$. 
Using the linear structure of $\MBQ^n$, we will obtain a upper bound for its reset threshold.

Let $\MCA = (Q, \Sigma, \delta)$ be an $n$-state automaton. We always assume that $Q = \{1,\ldots, n\}$.
Fix a subset $A$ of $\Sigma_0$ and denote $\delta(A^*)$ by $G$.

Firstly, let us recall some concepts in linear space.
Let $X \subseteq \MBQ^n$. 
The linear subspace \emph{spaned} by $X$, denoted by $\Span(X)$, is defined as
\[
	\Span(X) := \left\{\sum_{x \in X} c_x x: c_x \in \MBQ \right\}.
\]
The \emph{cone} generated by $X$, denoted $\cone(X)$, is the set
\[
	\cone(X) := \left\{\sum_{x \in X} c_x x: c_x \in \MBQ_{\ge 0} \right\}
\] 
where $\MBQ_{\ge 0}$ is the set of non-negative rationals.
Write $\langle \cdot, \cdot \rangle$ for the \emph{standard inner product} of $\MBQ^n$, that is the map such that $\langle x, y \rangle = \sum_{i = 1}^n x(i)y(i)$ for every $x,y \in \MBQ^n$.
The \emph{polar cone} of $X$, denoted $X^{\circ}$, is the set
\[
	X^{\circ} := \{y \in \MBQ^n : \langle x, y \rangle \le 0, \forall x \in X\}.
\]
If $X$ is a linear subspace of $\MBQ^n$, the \emph{orthogonal complement} of $X$, denoted $X^{\perp}$, is defined as
\[
	X^{\perp} = \{y : \langle x, y \rangle = 0, \forall x \in X\}.
\] 
It clearly holds $X^{\circ} = X^{\perp}$ in the case that $X$ is a linear subspace.

Next we will encode some information of an $n$-state synchronizing automaton into some objects in $\MBQ^n$.
For a subset $S \subseteq Q$, define $\chr_S$ to be the $(1 \times n)$-vector over $\MBQ$ such that its $i$-th coordinate is 
\[
	\chr_S(i) = \begin{cases}
	1 & \text{if $i \in S$,}\\
	0 & \text{otherwise.}
	\end{cases}
\]
For any $q \in Q$, to simplify notation, we write $\chr_q$ for $\chr_{\{q\}}$.
Let $w$ be an arbitrary word in $\Sigma^*$. Define $[w]$ to be the $(n \times n)$-matrix over $\MBQ$ such that $\chr_q [w] = \chr_{q.w^{-1}}$ for all $q \in Q$. It is clear that $\chr_S[w] = \chr_{S.w^{-1}}$ for all $S \subseteq Q$.
Define $\kk_w$ for the $(1 \times n)$-vector over $\MBQ$ such that its $i$-th coordinate is
\[
	\kk_w(i) = |i.w^{-1}| - 1.
\] 
For every $g \in G$, set $\kk_g$ to be $\kk_w$ and $[g]$ to be $[w]$, where $w \in \Sigma^*$ is an arbitrary word such that $\delta(w) = g$.

A sequence $(X_i)_{i  \ge 0}$ is called \emph{eventually constant} if there exists an integer $j \ge 0$ such that for all $k > j$, $X_k = X_j$.
For an eventually constant sequence $\mathcal{X} = (X_i)_{i \ge 0}$, the minimum integer $j$ such that for all $k > j$, $X_k = X_j$ is called \emph{transient length} of $\mathcal{X}$, denoted by $\transLen(\mathcal{X})$ and write $\limit \mathcal{X}$ for $X_{\transLen(\mathcal{X})}$. 
In the following, we will define some eventually constant sequences which play a crucial role in our proof. 

Define $\mathcal{T}(\MCA, A) = (T_i)_{i \ge 0}$ and $\mathcal{K}(\MCA, A) = (K_i)_{i \ge 0}$, 
to be the two sequences such that 
\[
	T_i := \left\{\kk_w : w \in (\Sigma \setminus \Sigma_0) A^{\le i} \right\} \quad \text{and} \quad K_i := \cone(T_i),
\]
for every $i \ge 0$.
To simplify notations, without ambiguity, we write $\mathcal{T}$ and $\mathcal{K}$ for $\mathcal{T}(\MCA, A)$ and $\mathcal{K}(\MCA, A)$, respectively.

Since $\langle A \rangle$ is a finite group, $\mathcal{T}, \mathcal{K}$ are eventually constant. 
Denote $\limit \mathcal{T}$ and $\limit \mathcal{K}$ by $T_{\infty}$ and $K_{\infty}$, respectively.
Observe that $K_{\infty} = \cone(T_{\infty})$ and then 
\begin{equation} \label{eq:index}
	\transLen(\mathcal{T}) \ge \transLen (\mathcal{K}).
\end{equation}

We begin with some elementary results. According to the definition, it clearly holds that 
\begin{equation} \label{eq:elem}
	|S.w^{-1}| - |S| = \langle \chr_S, \kk_w \rangle
\end{equation}
for every $S \subseteq Q$ and $w \in \Sigma^*$. As consequences, we have the following lemma.

\begin{lemma} \label{lem:elem}
	Let $S \subseteq Q$.
	\begin{enumerate}
		\item \label{lem:extendWord} The subset $S$ is extended by $w$ if and only if $\langle \chr_S, \kk_w \rangle > 0$.
		\item \label{lem:larger} If $\chr_S \in (K_{\infty})^{\circ}$ then $|S.a^{-1}| = |S|$ for all $a \in \Sigma$.
		\item \label{lem:sumzero} For every vector $x \in K_{\infty}$, $\sum_{i = 1}^n x(i) = \langle \chr_Q, x \rangle = 0$.
	\end{enumerate}	
\end{lemma}

We say that $\MCA$ is \emph{strongly connected} if for every two states $p,q$, there exists a word $w \in \Sigma^*$ such that $p.w = q$.

\begin{lemma} \label{lem:leavePolarCone}
	Assume that $\MCA = (Q,\Sigma,\delta)$ is synchronizing and strongly connected. 
	Let $S$ be a nonempty proper subset of $Q$.
	If $\chr_S \in (K_{\infty})^\circ$, then there exists a word $w \in \Sigma^*$ such that $\chr_{S.w^{-1}} \notin K^{\circ}$.
\end{lemma}

\begin{proof}
	Since $\MCA$ is synchronizing and strongly connected, there exists a word 
	\[
		u = u_1 u_2 \cdots u_t \in \Sigma^*
	\] 
	such that $S.u^{-1} = Q$.
	Since $|S| < |Q|$ and \cref{lem:elem}~\cref{lem:larger}, there exists an integer $t' < t$ such that $\chr_{S.v^{-1}} \notin (K_{\infty})^{\circ}$, where $v = u_1 u_2 \cdots u_{t'}$.
\end{proof}

Due to \cref{lem:leavePolarCone},  we can define $\ell(S)$ to be the length of a shortest word $w$ such that $\chr_{S.w^{-1}} \notin (K_{\infty})^{\circ}$ for every nonempty proper subset $S \subsetneq Q$.

\begin{proposition} \label{prop:ext}
	Assume that $\MCA$ is synchronizing and strongly connected.
	Let $S$ be a nonempty proper subset of $Q$.
	Then $S$ is $(\transLen(\MCK) + \ell(S) + 1)$-extensible.
\end{proposition}
\begin{proof}
	Let $k = \transLen(\MCK)$ and $\ell = \ell(S)$.
	By \cref{lem:leavePolarCone}, there exists an $\ell$-length word $w = (w_1, \ldots, w_{\ell}) \in \Sigma^*$ such that $\chr_{S} \notin (K_{\infty})^{\circ}$. Since 
	\[
		S.(w_2, \ldots, w_{\ell})^{-1} \in (K_{\infty})^{\circ},
	\]
	by \cref{lem:elem}~\cref{lem:larger}, $|S.w^{-1}| = |S|$.

	Let $P = S.w^{-1}$.
	Since $\chr_P \notin (K_{\infty})^{\circ}$, there exists a vector $x \in K_{\infty}$ such that $\langle x, \chr_P \rangle > 0$. Since $K_{\infty} = \cone(T_k)$, there exists a vector $y \in T_k$ such that $\langle y, \chr_P \rangle > 0$. By the definition of $T_k$, we can find a word $u \in (\Sigma \setminus \Sigma_0) A^{\le k}$ such that $y = \kk_u$. By \cref{lem:elem}~\cref{lem:extendWord}, $|P.u^{-1}| > |P| = |S|$.
	Hence, $uw$ extends $S$ and then $S$ is $(k + \ell + 1)$-extensible.
\end{proof}

If $\transLen(\MCK)$ and $\ell(S)$ can be bounded by a linear function of $n$, using \cref{prop:ExtensionMethod}, one can bound $\rt(\MCA)$ by a quadratic function of $n$. In general, it is hard to estimate $\transLen(\MCK)$ and $\ell(S)$ of an automaton. 
However, in the next section, we will establish some linear bounds for $\transLen(\MCK)$ with the assumption $\MCA \in \ST$ and $\Sigma = \Sigma_0 \cup \Sigma_1$.
And, in the rest of this section, we will establish the following linear bound for $\ell(S)$ in the case that $K_{\infty}$ is a linear subspace of $\MBQ^n$.

\begin{proposition} \label{prop:ell}
	Assume that $\MCA$ is synchronizing and strongly connected.
	If $K_{\infty}$ is a linear subspace of $\MBQ^n$, then $\ell(S) \le n - 1 - \Dim(K_{\infty})$ for every nonempty proper subset $S \subsetneq Q$.
\end{proposition}

Before proving \cref{prop:ell}, we show that ``$G$ is transitive'' implies ``$K_{\infty}$ is a linear subspace of $\MBQ^n$''.

\begin{lemma} \label{lem:transitivity}
	If $G$ is transitive, then $K_{\infty}$ is the linear subspace spanned by $T_{\infty}$.
\end{lemma}
\begin{proof}
	It is sufficient to prove $-x \in K_{\infty} = \cone(T_{\infty})$ for each $x \in T_{\infty}$.
	Take an arbitrary $x \in T_{\infty}$ and let
	\[
		y := \sum_{g \in G} x[g].
	\]
	It is clear that $y \in K_{\infty}$.
	Let $i$ and $j$ be two arbitrary integers in $\{1,\ldots, n\}$.
	Since $G$ is transitive, there exists $h \in G$ such that $i.h = j$.
	Note that $y[h](j) = y(i)$ and 
	\[
		y[h] = \sum_{g \in G} x[g][h] = \sum_{g \in G} x[g] = y.
	\]
	Then $y(i) = y(j)$. 
	By the arbitrarity of $i$ and $j$, it holds that $y = c \chr_Q$ for some $c \in \MBQ$.
	Since $y \in K_{\infty}$, by \cref{lem:elem}~\cref{lem:sumzero}, we have $\sum_{i = 1}^n y(i) = 0$. This implies $c = 0$ and then 
	\[
		-x  = \sum_{g \in G \setminus \{\id\}} x[g],
	\]
	where $\id$ is the identity map in $\Sym(Q)$.
	Since $x[g] \in T_{\infty}$ for all $g \in G$, we obtain that $-x \in K_{\infty}$ as wanted.
\end{proof}

Now, we go back to prove \cref{prop:ell}. The following dimension argument plays a crucial role in our proof.

\begin{lemma} \label{lem:LA}
	Let $L$ be a subspace of $\MBQ^n$ and a non-zero vector $x \in L$. If there exists a word $w \in \Sigma^*$ such that $x[w] \notin L$ then there exists a word $w' \in \Sigma^*$  such that $x[w'] \notin L$ and $|w'| \le \Dim(L)$.
\end{lemma}
\begin{proof}
	For every nonnegative integer $i$, define $L_i := \Span(xv : v \in \Sigma^{\le i})$.
	Observe that there exists a unique integer $j$ such that
	\[
		L_0 \subsetneq L_1 \subsetneq \cdots \subsetneq L_j = L_{j+1} = \cdots
	\]
	Let $t$ be the minimum integer such that $L_t \nsubseteq L$. Observing that $t \le j$, we have
	\[
		t  = \Dim(L_0) + t - 1 \le \Dim(L_{t-1}) \le \Dim(L).
	\]
	Then there exists a word $w'$ of length $\le \Dim(L)$ such that $x[w'] \notin L$.
\end{proof}

\begin{proof}[Proof of \cref{prop:ell}]
	Since $K_{\infty}$ is a linear subspace of $\MBQ^n$, it holds that $(K_{\infty})^{\circ} = (K_{\infty})^{\perp}$ is also a linear subspace of $\MBQ^n$.   
	Observe that $\chr_Q \in (K_{\infty})^{\circ}$. Then we can decompose $(K_{\infty})^{\circ}$ as $(K_{\infty})^{\circ} = V_0 \oplus V_1$, where 
	\[
		V_0 = \left\{x \in (K_{\infty})^{\circ} : \langle \chr_Q, \kk_w \rangle = 0 \right\} \quad \text{and} \quad V_1 = \Span(\chr_Q).
	\]
	For every subset $R \subseteq Q$, define $\mathbf{p}_R := \chr_R - \frac{|R|}{n}\chr_Q$.
	For each $R \subseteq Q$, observe that $\chr_R \in (K_{\infty})^{\circ}$ if and only if $\mathbf{p}_R \in V_0$.

	Let $S$ be a nonempty proper subset of $Q$ such that $\chr_S \in (K_{\infty})^{\circ}$.
	Since $\MCA$ is synchronizing and strongly connected, let $w'$ be a reset word such that $Q.w' \subseteq S$.
	Then 
	\[
		\mathbf{p}_S[w'] =  \chr_S [w'] - \frac{|S|}{n}\chr_Q [w'] = \left(1 - \frac{|S|}{n} \right)\chr_Q \notin V_0.
	\] 
	Let $w = va$ be a shortest word such that $\mathbf{p}_S[w] \notin V_0$. \cref{lem:LA} provides that the length of $w$ is at most $\Dim(V_0) = n - 1 -  \Dim(K_{\infty})$. 

	We will complete the proof by showing $\chr_{S.w^{-1}} \notin (K_{\infty})^{\circ}$.
	Note that $\mathbf{p}_S[v] = \chr_{S.v^{-1}} - \frac{|S|}{n}\chr_Q$. Since $\mathbf{p}_S[v], \chr_Q \in (K_{\infty})^{\circ}$, we have $\chr_{S.v^{-1}} \in (K_{\infty})^{\circ}$. By \cref{lem:elem}~\cref{lem:larger}, $|S| = |S.v^{-1}| = |S.w^{-1}|$.  Hence,
	\[
		\mathbf{p}_{S.w^{-1}} =  \chr_{S.w^{-1}}- \frac{|S.w^{-1}|}{n}\chr_Q = \chr_{S.w^{-1}}- \frac{|S|}{n}\chr_Q  = \mathbf{p}_S[w] \notin V_0
	\]
	which is equivalent to $\chr_{S.w^{-1}} \notin (K_{\infty})^{\circ}$. 
\end{proof}

Combining \cref{prop:ExtensionMethod,lem:transitivity,prop:ext,prop:ell}, we establish the following bound.

\begin{theorem} \label{thm:main}
	Let $\MCA = (Q, \Sigma, \delta) \in \ST$ be an $n$-state automaton.	Then $\rt(\MCA) \le 1 + (n-2)(n - \Dim(K_{\infty}) + \transLen(\MCK(\MCA, A)))$, where
	$A \subseteq \Sigma_0$ such that $\delta(A^*)$ is transitive. 
\end{theorem}

\begin{remark}
	Note that $\Dim(K_{\infty}) \ge 1$ and $\dd_{\MCA}(A) \ge \transLen(\MCK(\MCA, A))$.
	\begin{enumerate}
		\item \cref{thm:main} improves \cref{thm:Rystsov}.
		\item One of \cref{thm:Steinberg} and \cref{thm:main} cannot deduce the other one. If we only want to establish a quadratic upper bound for reset thresholds of a special class of automata, \cref{thm:main} may have more advantages.
	\end{enumerate}
\end{remark}

\section{Rystsov digraphs} \label{sec: RystsovDigraphs}

This section is divided into two parts:
\begin{itemize}
\item In \cref{sec:growth}, we establish some results for digraphs.
\item In \cref{sec: Automata}, we derive some directed graphs from automata. Using the results in \cref{sec:growth} and \cref{thm:main}, we prove \cref{thm:transitive}.
\end{itemize}

\subsection{Digraphs} \label{sec:growth}

Firstly, we recall some notations of digraphs.
A \emph{digraph} $\Gamma = (V,E)$ is an ordered pair of sets such that $E \subseteq V \times V$. The set $V$ is called the \emph{vertex set} of $\Gamma$ and the $E$ is called the \emph{arc set} of $\Gamma$. 
We assume that the digraphs in this section are \emph{loop-free}, that is, $(v,v)$ is not an arc for every vertex $v$.
Let $u$ and $v$ be two vertices of $\Gamma$. 
A sequence of vertices $(v_1, \ldots, v_t)$ is called a \emph{path} from $u$ to $v$ if $v_1 = u$, $v_t = v$ and $(v_i, v_{i+1}) \in E$ for every $1 \le i \le t -1$.
We say that $u$ and $v$ are \emph{connected} if there exists a sequence of vertices $(v_1, \ldots, v_t)$ such that $v_1 = u$, $v_t = v$ and either $(v_i, v_{i+1}) \in E$ or $(v_{i+1}, v_i) \in E$ for every $1 \le i \le t -1$.
The \emph{strong connectivity} of $\Gamma$, denoted $\strc(\Gamma)$, is the equivalence relation such that $(u,v) \in \strc(\Gamma)$ if and only if there exist a path from $u$ to $v$ and a path from $v$ to $u$.
The \emph{weak connectivity} of $\Gamma$, denoted $\weakc(\Gamma)$, is the equivalence relation such that $(u,v) \in \weakc(\Gamma)$ if and only if $u,v$ are connected.
Equivalence classes of $\strc(\Gamma)$ ($\weakc(\Gamma)$ resp.,) are called a \emph{strongly connected components} (\emph{weakly connected components} resp.,) of $\Gamma$.
Write $\Strong(\Gamma)$ ($\Weak(\Gamma)$, $\Sink(\Gamma)$, resp.) for the set of strongly connected components (weakly connected components, sink components, resp.) of $\Gamma$.
A strongly connected component $C$  of $\Gamma$ is called a \emph{sink component} of $\Gamma$ if there is no arc $(u,v) \in E$ such that $u \in C$ and $v \notin C$; is called a \emph{source component} of $\Gamma$ if there is no arc $(u,v) \in E$ such that $u \notin C$ and $v \in C$.
A strongly connected component is \emph{non-sink} (\emph{non-source}, resp.) if it is not a sink (source, resp.) component of $\Gamma$.

Let $V = \{1, \ldots, n\}$ and $A \subseteq \Sym(V)$.
We will say that the sequence $\Gamma_0, \Gamma_1, \ldots$ is an \emph{$A$-growth} of $\Gamma_0$ if
	\begin{itemize}
		\item $\Gamma_0 = (V, E_0)$ is a digraph with vertex set $V$;
		\item for every positive integer $i$, $\Gamma_i = (V,E_i)$ is the digraph such that $E_i = \{(p.w, q.w) : (p,q)\in E_0, w \in A^{\le i}\}$.
	\end{itemize}
This concept was firstly used in \cite{Rys00}, and appear widely in the research of synchronizing automata \cite{BCV23,BV16,GGJGV19,Rys00,RS24}.

In the rest of \cref{sec:growth}, we set $\Gamma_0, \Gamma_1, \ldots$ is a $A$-growth of $\Gamma_0$.
Write $G = \langle A \rangle$. Since $G$ is finite, the sequence $\Gamma_0, \Gamma_1, \ldots$ is eventually constant. Denote $\limit(\Gamma_0, \Gamma_1, \ldots)$ by $\Gamma_{\infty}$.

\begin{lemma} \label{lem:graph}
	If $G$ is transitive, then $\weakc(\Gamma_{\infty}) = \strc(\Gamma_{\infty})$.
\end{lemma}
\begin{proof}
	It follows the fact that $\Gamma_{\infty}$ is a finite vertex-transitive digraph.
\end{proof}

For all $(p,q) \in V \times V$, the \emph{incidence vector} of $(p,q)$ is the vector $\chr_p - \chr_q$, denoted $\chi_{(p,q)}$. 
For a digraph $\Gamma = (V,E)$, write $L(\Gamma)$ for the subspace $\Span(\chi_e : e \in E)$ of $\MBQ^n$.
The following result is well-known in algebraic  graph theory (see \cite[Chapter 4]{Biggs74}). 
\begin{lemma} \label{lem:translate}
	Assume that $W_1, \ldots, W_d$ are all weakly connected components of a digraph $\Gamma = (V,E)$. 
	Then the subspace $(L(\Gamma))^{\perp}$ is the $d$-dimensional subspace $\Span (\chr_{W_1}, \ldots, \chr_{W_d})$.
\end{lemma}

\begin{lemma} \label{lem:weaklyconn}
	Let $d = |\Strong(\Gamma_{\infty})|$.
	If the arc set of $\Gamma_0$ is nonempty and $G$ is transitive, then $\weakc(\Gamma_{n - d - 1}) = \weakc(\Gamma_{\infty}) = \strc(\Gamma_{\infty})$.
\end{lemma}
\begin{proof}
	By \cref{lem:graph}, $\weakc(\Gamma_{\infty}) = \strc(\Gamma_{\infty})$.
	It is sufficient to show $\weakc(\Gamma_{n - d - 1}) = \weakc(\Gamma_{\infty})$.

	Define  $\mathcal{L} = (L_i)_{i \ge 0}$  to be the sequence of linear spaces, where $L_i := L(\Gamma_i)$ for every $i \ge 0$. Write $L_{\infty} = \limit \mathcal{L}$.
	By \cref{lem:translate}, $\Dim((L_{\infty})^{\perp})= d$ and then $\Dim(L_{\infty}) = n - d$.
	For every $i < \transLen(\mathcal{L})$, since $L_i \subsetneq L_{i+1}$, we have $\Dim(L_i) < \Dim(L_{i+1})$.
	Due to $E_0 \neq \emptyset$, it holds that $\Dim(L_0) \ge 1$ and then $L_{n - d - 1} = L_{\infty}$.
	By \cref{lem:translate},  $\weakc(\Gamma_{n - d - 1}) = \weakc(\Gamma_{\infty})$.
\end{proof}

\begin{lemma} \label{lem:degree}
	Assume that $E_0 \neq \emptyset$ and $G$ is transitive.
	For every $v \in V$, there exist two vertices $p,q \in V$ such that $(v,p), (q,v) \in E_{n-1}$.
\end{lemma}
\begin{proof}
	Since $E_0 \neq \emptyset$, let $x,y$ be two vertices such that $(x,y) \in E_0$.
	Let $v$ be an arbitrary vertex in $V$.
	Since $G$ is transitive, there exist two words $w, w' \in A^{\le n-1}$ such that $x.w = y.w' = v$.
	Then $(v, y.w), (x.w', v) \in E_{n-1}$.
\end{proof}

\begin{lemma} \label{lem:nonSink}
	Let $i$ be a positive integer. If $\strc(\Gamma_{i + 1}) = \strc(\Gamma_i) \neq \strc(\Gamma_{\infty})$ and $G$ is transitive,
	then there exists a non-sink strongly connected component $C$ of $\Gamma_i$ and $a \in A$ such that $C.a$ is a sink component of $\Gamma_i$.
\end{lemma}
\begin{proof}


	Assume, for a contradiction, for all non-sink strongly connected component $C$ and $a \in A$ such that $C.a \in  \Strong(\Gamma_i) \setminus \Sink(\Gamma_i)$.
	Since $\strc(\Gamma_{i + 1}) = \strc(\Gamma_i)$, $X . a \in \Strong(\Gamma_i)$ for each $X \in \Strong(\Gamma_i)$ and $a \in A$.
	For every $a \in A$, let $a'$ be the permutation on $\Strong(\Gamma_i)$ such that $X.a' = X.a$ for any $X \in \Strong(\Gamma_i)$.
	Let $A'$ be the set $\{a' : a \in A\}$.
	Since $G = \langle A \rangle$ is transitive on $V$, the permutation group $\langle A' \rangle$ is transitive on $\Strong(\Gamma_i)$.
	Then there exist $C \in \Strong(\Gamma_i) \setminus \Sink(\Gamma_i)$ and $a \in A$ such that $C.a \in \Sink(\Gamma_i)$.

	
\end{proof}


\begin{lemma} \label{lem:decrease}
	Let $i$ be a non-negative integer such that $\strc(\Gamma_i) \neq \strc(\Gamma_{\infty})$. 
	If $G$ is transitive, then either 
	\[
		|\Strong(\Gamma_i)| > |\Strong(\Gamma_{i+1})|
	\]
	or 
	\[
		|\Sink(\Gamma_i)| > |\Sink(\Gamma_{i+1})|.
	\]
\end{lemma}

\begin{proof}
	If $|\Strong(\Gamma_i)| > |\Strong(\Gamma_{i+1})|$, we are done.
	Otherwise, since $\strc(\Gamma_i) \subseteq \strc(\Gamma_{i+1})$, we have $\strc(\Gamma_i) = \strc(\Gamma_{i+1})$ and 
	\begin{equation} \label{eq:sink}
		\Sink(\Gamma_i) \supseteq \Sink(\Gamma_{i+1}).
	\end{equation}
	
	By \cref{lem:nonSink}, there exist $C \in \Strong(\Gamma_i) \setminus \Sink(\Gamma_i)$ and $a \in A$ such that $C.a \in \Sink(\Gamma_i)$.
	Since $C \in \Strong(\Gamma_i) \setminus \Sink(\Gamma_i)$, there exists an arc $(p,q) \in E_i$ such that $p \in C$ and $q \notin C$.
	Note that $(p.a, q.a) \in E_{i+1}$. Since $p.a \in C.a$ and $q.a \notin C.a$, it holds that $C.a$ is a non-sink component of $\Gamma_{i+1}$. 
	By \cref{eq:sink}, we have $|\Sink(\Gamma_i)| > |\Sink(\Gamma_{i+1})|$.
\end{proof}

\begin{lemma} \label{lem:equalTrivial}
	Let $d = |\Strong(\Gamma_{\infty})|$.
	If $G$ is transitive and $d > \frac{n}{3}$, then $\strc(\Gamma_{n}) = \strc(\Gamma_{\infty})$. 
\end{lemma}
\begin{proof}
	Since $G$ is transitive, $d$ divides $n$. Then $d \in \{\frac{n}{2}, n\}$.
	Since each strongly connected component of $\Gamma_{\infty}$ has at most $2$ vertices, it is clear that there exist at most $n$ arcs in $\Gamma_{\infty}$. 
	Note that $|E_i| < |E_{i+1}|$ for every integer $i$ such that $\Gamma_i \neq \Gamma_{\infty}$.
	Then $\strc(\Gamma_{n}) = \strc(\Gamma_{\infty})$.
\end{proof}

\begin{lemma} \label{lem:equal}
	Let $d = |\Strong(\Gamma_{\infty})|$.
	If $G$ is transitive and $d \le \frac{n}{3}$, then $\strc(\Gamma_{2n-3d-1}) = \strc(\Gamma_{\infty})$. 
\end{lemma}
\begin{proof}

	Let $m$ be the minimum integer such that $\strc(\Gamma_m) = \strc(\Gamma_{\infty})$.
	Define 
	\[
		f(i) = |\Strong(\Gamma_i)| + |\Sink(\Gamma_i)|,
	\] 
	for each $i \ge 0$.

	Consider $\Gamma_{n-1}$.
	By \cref{lem:weaklyconn}, $\Gamma_{n-1}$ has $d$ weakly connected components and $\weakc(\Gamma_{n-1}) = \strc(\Gamma_m)$.
	In the case that $m \le n-1$, then $m \le 2n - 3d - 1$ and we are done. 

	Now we assume that $m > n - 1$.
	Let $C$ be a weakly connected component of $\Gamma_{n-1}$ but not a strongly connected component of $\Gamma_{n-1}$.
	Define 
	\begin{align*}
		a &:= \text{the number of source components of $\Gamma_{n-1}$ in $C$}; \\
		b &:= \text{the number of non-source non-sink components of $\Gamma_{n-1}$ in $C$}; \\
		c &:= \text{the number of sink components of $\Gamma_{n-1}$ in $C$}.
	\end{align*}

	Let $D$ be either a source component or a sink component of $\Gamma_{n-1}$ in $C$. 
	By \cref{lem:degree}, there exists an arc in $D$ and then there are at least two vertices in $D$. 
	This implies 
	\[
		2a + b + 2c \le |C| = \frac{n}{d}.
	\]
	Since $a \ge 1$, we have 
	\begin{equation} \label{eq:counting}
		(a + b + c) + c \le \frac{n}{d} - 1.
	\end{equation}
	The left hand side of \cref{eq:counting} is the sum of the number of strongly connected components and the number of sink components of $\Gamma_{n-1}$ in $C$. 
	Since $d \le \frac{n}{3}$, we have
	\begin{align*} \label{eq:f}
		f(n-1)  &\le (\frac{n}{d} - 1)x + 2(d-x)  \tag{by \cref{eq:counting}} \\ 
		 &\le  (\frac{n}{d} - 1)d  \tag{by $n/d - 1 \ge 2$}\\
		 &= n - d, 
	\end{align*}
	where $x = |\Weak(\Gamma_{n-1}) \setminus \Strong(\Gamma_{n-1})|$.

	Since $f(m) = 2d$, using \cref{lem:decrease}, we have
	\[
		m - (n-1) \le f(n-1) - f(m) \le (n - d) - 2d = n - 3d.
	\]
	Hence, $m \le 2n - 3d - 1$.
\end{proof}

\begin{remark}
	\begin{enumerate}
		\item In the case that $\Gamma_{\infty}$ is strongly connected and $n \ge 3$, \cref{lem:equal} shows that $\strc(\Gamma_{2n- 4}) = \strc(\Gamma_{\infty})$.
		This slightly improves \cite[Theorem 2]{Rys00}, \cite[Lemma 6]{GGJGV19} and \cite[Theorem 2]{RS24}. And then one can slightly improve the bounds of reset thresholds in \cite[Theorem 3]{Rys00}, \cite[Theorem 7]{GGJGV19} and \cite[Theorem 4]{RS24}.
		\item In \cite[Section 3]{GGJGV19}, for every odd integer $n \ge 7$, Gonze, Gusev, Gerencs\'{e}r, Jungers and Volkov constructed two permutations $a,b \in \Sym(n)$ such that 
		\begin{itemize}
			\item $\langle a,b \rangle$ is $2$-homogeneous;
			\item for any word $w \in \{a,b\}^*$, if $\{2,4\}.w = \{\frac{n-1}{2}, \frac{n+3}{2}\}$, then $|w| \ge \frac{n^2}{4} + O(n)$. 
		\end{itemize}
		Let $\Gamma_0 = (V, E)$ be a digraph such that $V = \{1, \ldots, n\}$ and $E =\{(2,4)\}$. Let $\Gamma_0, \Gamma_1, \ldots$ be the $\{a,b\}$-growth of $\Gamma_0$. 
		Observe that if $\Gamma_m = \Gamma_{\infty}$ then $m \ge \frac{n^2}{4} + O(n)$, although $\strc(\Gamma_{2n-4}) = \strc(\Gamma_{\infty})$.

		Meanwhile, the above two permutations also provide a negative answer for \cite[Problem 12.39]{ACS17}.
	\end{enumerate}
\end{remark}

\subsection{Automata} \label{sec: Automata}

In this subsection, we will define a sequence of digraphs with respect to an automaton $\MCA = (Q , \Sigma, \delta)$ and $A \subseteq \Sigma_0$.

For a $1$-defect word $w \in \Sigma^*$, the \emph{excluded state} of $w$ is the state such that $\excl(w) \notin Q.w$, denoted $\excl(w)$; the \emph{duplicate state} is the state $q$ such that $|q.w^{-1}| > 1$, denoted $\dupl(w)$.

For $i \ge 0$, define $\Gamma_i := (Q, E_i)$ to be the digraph where
\[
	E_i := \left\{ (\excl (w),\dupl(w)) : w \in \Sigma_1 A^{\le i} \right\}.
\]

\begin{lemma} \label{lem:group}
	For all $i \ge 0$ and $a \in \Sigma_0$, if $(p,q) \in E_i$, then $(p.a ,q.a)\in E_{i+1}$. In particular, the sequence $ (\Gamma_0, \Gamma_1, \ldots)$ is the $\delta(A)$-growth of $\Gamma_0$.
\end{lemma}
\begin{proof}
	Let $w \in \Sigma_1 \Sigma_0^i$ such that $p = \excl(w)$ and $q = \dupl(w)$.
	By directly computing, we have $\excl(wa) = \excl(w).a$ and $\dupl(wa) = \dupl(w).a$ for all $a \in \Sigma_0$. Then $(p.a ,q.a)\in E_{i+1}$.
	Hence $ (\Gamma_0, \Gamma_1, \ldots)$ is the $\delta(A)$-growth of $\Gamma_0$.
\end{proof}

\begin{proposition} \label{prop:constInx}
	Assume that $\MCA = (Q, \Sigma, \delta) \in \ST$.
	If $K_{\infty}$ is a linear subspace of $\MBQ^n$ and $\Sigma = \Sigma_0 \cup \Sigma_1$, then 
	\[
		\transLen(\MCK) \le \begin{cases}
			n & \text{if $\Dim(K_{\infty}) = \frac{n}{2}$,} \\
			3\Dim(K_{\infty}) - n - 1 & \text{otherwise}. \\
		\end{cases}
	\]
\end{proposition}
\begin{proof}
	Recall the definition of $T_i$ that $T_i := \{\kk_w : w \in \Sigma_{\ge 1} A^{\le i}\}$ for $i \ge 0$. 
	Since $\Sigma = \Sigma_0 \cup \Sigma_1$, we have $T_i = \{\chi_e : e \in E_i\}$ for $i \ge 0$.
	Let $m$ be the minimal integer such that $\strc(\Gamma_m) = \strc(\Gamma_{\infty})$.
	By \cref{lem:graph}, $\strc(\Gamma_{\infty}) = \weakc(\Gamma_{\infty})$ and $\strc(\Gamma_m) = \weakc(\Gamma_m)$
	This implies that $\Span(T_{\infty}) = \cone(T_{\infty})$ and $\Span(T_m) = \cone(T_m)$.
	Let $C_1, \ldots, C_d$ be the strongly connected components of $\Gamma_m$.
	Since $\delta(\Sigma_0)$ is transitive, using \cref{lem:translate}, 
	\[
		\Span(C_1, \ldots, C_d) = \Span(T_m)^{\perp} = \cone(T_m)^{\perp} = K_m^{\circ}
	\]
	and
	\[
		\Span(C_1, \ldots, C_d) = \Span(T_\infty)^{\perp}  = \cone(T_{\infty})^{\perp} = K_{\infty}^{\circ}.
	\]
	Then $K_{m} = K_{\infty}$ and $\dim(K_{\infty}) = n - d$.
	By \cref{lem:equalTrivial,lem:equal}, 
	\[
		\transLen(\MCK) \le \begin{cases}
			n & \text{if $\Dim(K_{\infty}) = \frac{n}{2}$,} \\
			3\Dim(K_{\infty}) - n - 1 & \text{otherwise}. \\
		\end{cases}
	\]
\end{proof}

\begin{proof}[Proof of \cref{thm:transitive}]
	Computer experiments confirmed \v{C}ern\'{y} conjecture for any synchronzing automata with at most 5 states (see \cite[Table 2]{KKS16}).
	One can check directly that $\rt(\MCA) \le (n-1)^2 \le n^2 - 7n + 7$ for $n \le 5$.

	Now, we assume that $n \ge 6$.
	Using \cref{lem:translate,lem:graph}, $\frac{n}{2} \le \Dim(K_{\infty}) \le n-1$.
	Let $S$ be a nonempty proper subset of $Q$.
	By \cref{prop:ell,prop:constInx},
	if $\Dim(K_{\infty}) = \frac{n}{2}$ then 
	\begin{align} 
		\begin{split}\label{eq:extensionbound1}
			\transLen(\MCK) + \ell(S) + 1 &\le  n + (n - 1 - \frac{n}{2}) + 1 \\
			 & = 2n - 3 - (\frac{n}{2} - 3) \le 2n - 3;
		\end{split}
	\end{align}
	if $\Dim(K_{\infty}) > \frac{n}{2}$, 
	\begin{align}
		\begin{split}\label{eq:extensionbound2}
		\transLen(\MCK) + \ell(S) + 1 &\le (3\Dim(K_{\infty}) - n - 1) + (n - 1 - \Dim(K_{\infty})) + 1 \\
		&= 2\Dim(K_{\infty}) - 1 \le 2(n-1) - 1 =  2n-3.
		\end{split}
	\end{align}
	Combining \cref{prop:ext,eq:extensionbound1,eq:extensionbound2},  every nonempty proper subset of $Q$ is $(2n - 3)$-extensible.
	Using \cref{prop:ExtensionMethod}, we obtain that
	\[
		\rt(\MCA) \le 1 + (n-2)(2n-3) = 2n^2 - 7n + 7.
	\]
\end{proof}

\section{Conclusion} \label{sec:conclusion}

We obtain an upper bound for the reset thresholds of $\ST$-automata which improves Rystsov's bound. Using this upper bound, we prove that there exists a synchronzing word of length at most $n^2 - 7n + 7$ for every synchonizing $n$-state $\ST$-automata whose letters of defect at most $1$.


    \newcommand{\etalchar}[1]{$^{#1}$}

\end{document}